\documentclass[11pt]{article}
\usepackage[centertags]{amsmath}
\usepackage{amssymb}
\usepackage{amsthm}
\usepackage{hyperref}
\usepackage{bm}

\newtheorem{thm}{Theorem}[section]
\newtheorem{cor}[thm]{Corollary}

\newtheorem{prop}[thm]{Proposition}
\newtheorem{defn}[thm]{Definition}
\newtheorem{rem}[thm]{Remark}

\numberwithin{equation}{section}

\newcommand{\condexp}[1]{\left|}
\newcommand{\Real}{\mathbb R}

\newcommand{\punt}{\boldsymbol{.}}

\newcommand{\ibs}{\boldsymbol{i}}
\newcommand{\Nbs}{\boldsymbol{N}}

\newcommand{\iotabs}{\boldsymbol{\iota}}
\newcommand{\nubs}{\boldsymbol{\nu}}

\newcommand{\psibs}{\boldsymbol{\psi}}
\newcommand{\ubs}{\boldsymbol{u}}
\newcommand{\sr}{\scriptscriptstyle [p]}

\newcommand{\sY}{\scriptscriptstyle{[P]}}
\newcommand{\chibs}{\boldsymbol{\chi}}
\newcommand{\jbs}{\boldsymbol{j}}

\newcommand{\Ibs}{\boldsymbol{I}}
\newcommand{\wbs}{\boldsymbol{w}}

\newcommand{\mbs}{\boldsymbol{m}}
\newcommand{\zbs}{\boldsymbol{z}}

\newcommand{\kbs}{\boldsymbol{k}}

\newcommand{\ybs}{\boldsymbol{y}}
\newcommand{\xbs}{\boldsymbol{x}}
\newcommand{\ml}{\mathfrak{m}}
\newcommand{\Xbs}{\boldsymbol{X}}

\newcommand{\Tr}{\hbox{\rm Tr}}
\newcommand{\lambdabs}{\boldsymbol{\lambda}}

\newcommand{\mmodels}{\boldsymbol{\vdash}}

\begin{document}
 
\title{On photon statistics parametrized by a non-central Wishart random matrix}
\author{E. Di Nardo \thanks{Department of Mathematics, Computer science and Economics,
University of Basilicata, Viale dell'Ateneo
Lucano 10, 85100 Potenza, Italia; elvira.dinardo@unibas.it}}
\date{}
\maketitle
\begin{abstract}
In order to tackle parameter estimation of photocounting distributions, polykays of acting intensities are proposed as a new tool for computing photon statistics. As unbiased estimators of cumulants, polykays are computationally feasible thanks to a symbolic method recently developed in dealing with sequences of moments. This method includes the so-called method of moments for random matrices and results to be particularly suited to deal with convolutions or random summations of random vectors. The overall photocounting effect on a deterministic number of pixels is introduced. A random number of pixels is also considered. The role played by spectral statistics of random matrices is highlighted in approximating the overall photocounting distribution when acting intensities are modeled by a non-central Wishart random matrix.  Generalized complete Bell polynomials are used in order to compute joint moments and joint cumulants of multivariate photocounters. Multivariate polykays can be successfully employed in order to approximate the multivariate Poisson-Mandel transform. Open problems are addressed at the end of the paper.
\end{abstract}



\section{Introduction}
The recent renewed interest in multivariate photocounting originates within astronomical literature in connection with extrasolar planet detection methods \cite{Figer} or, more in general, in speakle patterns produced by direct imaging \cite{Goodman2}. A photovent occurs when light striking a pixel causes one or more electrons to be ejected. Within fixed time intervals $T,$ multivariate photocounters are modeled by non-negative random vectors counting photoevents from a set of $d$ pixels. Their stochastic fluctuations depend on intensities of light, encoded in a vector $\Ibs =(I_1, \ldots,I_d),$ with joint distribution $\mu({\rm d} \, \Ibs)$ on $(\Real^+)^{d}.$ The photocounters $\{N_i, i=1,\ldots,d\}$ are assumed conditionally independent and distributed according to a Poisson law parametrized by $\Ibs.$ The multivariate Poisson-Mandel transform \cite{Ferrari}
\begin{equation}
{\mathbb P}(\Nbs=\kbs)=\int \cdots \int_{(\Real^+)^{d}} \prod_{j=1}^d \frac{(I_j)^{k_j}}{{k_j}!} \, \exp(- I_j) \, \mu({\rm d} \, \Ibs)
\label{(2)}
\end{equation}
gives the joint distribution of $\Nbs=(N_1, \ldots, N_d)$ when $\kbs = (k_1, \ldots, k_d)$ is a multi-index of non-negative integers.
The availability of closed form formulae for (\ref{(2)}) depends upon the stochastic model $\mu.$ For a wide range of univariate stochastic intensities, these formulae have tractable expressions and sufficient statistics can be recovered by means of likelihood methods \cite{Chatelain}. In \cite{Wu}, an overview is given on the methods available to analyze the data sampled in photocounting, as for example the photon counting histogram (PCH). Fluorescence cumulant analysis (FCA) is indicated as the first theory that describes the effect of sampling time by measuring the spontaneous intensity fluctuations of fluorescent molecules \cite{Muller}. This measurement is done by
using factorial cumulants. Let us recall that if $M(t)$ denotes the moment generating function of a random variable (r.v.), then cumulants are the coefficients of $\log(M(t))$ and factorial cumulants are the coefficients of  $\log(M(\log(1+t))).$ Cumulants have special properties and if some of these properties can be deduced from data, then special stochastic models can be inferred \cite{Mordovina}. For example, if conditioned cumulants linearize, then the recorded data can be modeled by using (\ref{(2)}). The same relationship holds for factorial moments.

Inspired by FCA, the main goal of this paper is to extend cumulant analysis to the multivariate case,  which is still a challenging problem in particular when $\Ibs$ are entries of a random matrix \cite{Ferrari}. The preliminary contributions given in the literature  indicate factorial moments as simpler expressions to be used \cite{Chatelain1,Tourneret}.

In working with random matrices, the method of moments is still extensively used \cite{Debashis}. In general, some conditions on moments or moment generating functions need to be required and  applications of this approach to random matrices rely on some advanced combinatorial tools as for example zonal polynomials or hypergeometric functions \cite{Saw, Shah, Waal}. A first way to overcome this drawback is free probability, a non-commutative theory of probability  \cite{Speicher}. Indeed, random matrices are non-commutative objects whose large-dimension asymptotic are usually analyzed by using free probability. However, there are some aspects of random matrix theory to which the tools of free probability are not sufficient by themselves to resolve.

In this paper, we propose the employment of a method, called the symbolic method of moments \cite{Dinardo3}, which can be considered the commutative counterpart of free probability. In the symbolic method of moments, a sequence of numbers is represented by a symbol, called umbra, through a linear operator, sharing many properties of expectation. The elements of the sequence are called {\it moments}. This symbolic approach overcomes the well-known moment problem, since a sequence of numbers is dealt as it was a sequence of moments with no reference to any probability space as within free probability. It is mainly a tool to perform computations: the matching with r.v.'s is done a-posteriori. In difference from free probability,  just one operator is employed within the symbolic method of moments but the same sequence of moments may correspond to more than one symbol. For a pair of commutative r.v.'s, freeness\footnote{In free probability, free r.v.'s correspond to classical independent r.v.'s.} is equivalent to claim that at least one of them has vanishing variance. So freeness is a pure not-commutative device  that is why the symbolic method of moments may be considered an analogous of free probability in a commutative field. One of the strengths of this method is that questions on convergence of moment generating functions may be discarded and only polynomials involving umbrae are employed.

That having been said, the novelty of this paper is twofold. The first has implications beyond photocounting and involves the employment of the symbolic method of moments within random matrices. The second novelty is strictly related to photocounters and their statistics. In FCA, factorial cumulants are employed and calculated from the moments of the recorded photon counts. Quite recently efficient algorithms \cite{Dinardo} have been developed in order to compute experimental measurements of cumulants, known in the literature as polykays. The sampling behavior of polykays is much simpler than sample moments \cite{McCullagh} but their employment was not so widespread in the statistical community, due to the past computational complexity in recovering their expression. When complex amplitudes of incoherent waves have independent circular Gaussian distribution, the resulting acting intensity $\Ibs$ is the diagonal of a non-central Wishart random matrix and we show how spectral polykays \cite{Dinardo4} can be fruitfully employed in order to estimate their cumulants and then factorial cumulants of photocounting.

We prove that photocounter cumulants have a simpler expression compared with moments and factorial moments, taking advantage from the plainness of cumulants of the non-central Wishart random matrix. Thanks to the symbolic method of moments, the generalization to multivariate framework is straightforward. And this is an additional novelty of the paper since far fewer results can be found in the literature on compound Poisson random vectors \cite{O'Connor}, due to the difficulty of managing their distributions.

Since the symbolic method of moments is a new tool for photocounting, we get the opportunity of introducing this theory by dealing with a new measure describing the overall photocounting effect, together with its generalization to superposition of a random number of incoherent waves. New formulae are proposed to perform all these computations. Implementations of these formulae in {\tt Maple}
are available on demand.

The paper is organized as follows: the symbolic method of moments is sketched  in Section 3 for the univariate case and in Section 4 for the multivariate case. Section 2 shows how to model acting intensities when complex amplitudes of incoherent waves have independent circular Gaussian distribution. Polykays and their properties are recalled in Section 3. Formulae for computing moments, factorial moments and factorial cumulants of multivariate photocounters are given in Section 4. A series expansion of the multivariate Poisson-Mandel transform shows how
to by-pass the multi-dimensional integral in \eqref{(2)} and to use  multivariate polykays for numerical approximations. Some open problems are suggested at the end of the paper.
\section{Acting intensities modeled by Wishart random matrices}
The complex amplitude $\psi(x,y)$ of a wave can be modeled as $\psi(x,y)=m(x,y) +  X(x,y)$ \cite{Aime}, where $m(x,y)$ is a deterministic term
proportional to the wave amplitude without turbulence, and $X(x,y)$ is a random term distributed according to a zero mean complex Gaussian distribution.
This choice depends on the central limit theorem, since $X(x,y)$ represents the uncorrected part of the wave amplitude caused by errors.
The instantaneous intensity is defined as $I(x,y) = |X(x,y) + m(x,y)|^2.$ When $p$ incoherent waves are considered, their superposition field intensity $I_{\sr}(x,y)$ is obtained by summing the intensities of each wave, that is $I_{\sr}(x,y) = \sum_{i=1}^p |X_i(x,y) + m_i(x,y)|^2.$
A multidimensional framework is necessary when $d$ pixels are involved and correlations among wave amplitudes at $d$ different positions are encoded in a full rank covariance matrix $\Sigma.$ Then the complex amplitude of the $i$-th incoherent wave at $(x,y)$ is given by
$\psibs_i(x,y)=\mbs_i(x,y) + \Xbs_i(x,y),$ where the deterministic term $\mbs_i(x,y)$ is a $d$-dimensional vector and the random term
$\Xbs_i(x,y)$ is a $d$-variate circular complex Gaussian random vector with zero mean and full rank Hermitian covariance matrix $\Sigma.$
By omitting the notation $(x,y)$ for brevity, the resulting vector of intensities $\Ibs_{\sr}=(I_{1,\sr}, \ldots, I_{d,\sr})$ has components
\begin{equation}
I_{j,\sr} = \sum_{i=1}^p |X_{i j} + m_{i j}|^2 \qquad \hbox{for} \,\, j=1,2, \ldots, d
\label{(4)}
\end{equation}
with $X_{i j } = (\Xbs_i)_{j}$ and $m_{i j} = (\mbs_i)_{j}.$ The elements of $\Ibs_{\sr}$ are on the diagonal of a non-central Wishart square random matrix of order $d$
\begin{equation}
W_d(p, \Sigma, M) = \sum_{i=1}^p (\Xbs_{i} + \mbs_{i})^{\dag} (\Xbs_{i} + \mbs_{i}) \qquad \hbox{with} \,\, M = \sum_{i=1}^p \mbs_{i}^{\dag} \mbs_{i},
\label{(3)}
\end{equation}
usually denoted by $W_d(p).$ In equation \eqref{(3)}, $\dag$ denotes the conjugate transpose.  In the literature
$\Omega = \Sigma^{-1} M$ is called the non-centrality matrix.

The intensity $\Ibs_{\sr}$ is the parameter of a photocounting vector $\Nbs^{\sr}=(N_{1,\sr},\ldots,N_{d,\sr})$ with Poisson distribution
\eqref{(2)}. The Poisson r.v.'s $N_{1,\sr},\ldots,N_{d,\sr}$ are conditionally independent, that is $P\left(\Nbs^{\sr} =\kbs \, | \, \Ibs_{\sr} \right) = \prod_{j=1}^d P \left( N_{j,\sr}  = k_j \, | \, I_{j,\sr} \right).$

By recalling that complex Gaussian random vectors of dimension $d$ can be replaced by real Gaussian random vectors of dimension $2d,$ whose first $d$-components represent the real part and the last $d$-components represent the imaginary part, an alternative way to compute $I_{j,\sr}$ in \eqref{(4)} is $I_{j,\sr} = \sum_{i=1}^p \left[ {\rm Re}(X_{ij} + m_{ij})^2 +  {\rm Im}(X_{ij} + m_{ij})^2 \right],$
for $j=1,2, \ldots, d.$ However, this formula picks out elements on the diagonal of a real non-central Wishart random matrix of dimension $2 \, d$ and therefore is less efficient. Then,  we refer to complex non-central Wishart random matrices \eqref{(3)}.
\section{Overall photocounters}
Let us introduce the notion of overall photocounter.
\begin{defn} \label{photo}
The overall photocounter is ${\mathcal N}_{\sr}=N_{1,\sr} + \cdots + N_{d,\sr}$ if $p$ incoherent waves hit $d$ pixels and $N_{j,\sr}$ denotes the number of photoevents of the $j$-th pixel labeled with $j=1,\ldots,d.$
\end{defn}
Since the convolution of two (or more) mixed Poisson distributions is itself a mixed Poisson distribution,
with mixing densities the convolution of the two (or more) mixed Poisson distributions \cite{Ferrari}, then
the overall photocounter ${\mathcal N}_{\sr}$ has a mixed Poisson distribution with random parameter $I_{1,\sr} + \cdots + I_{d,\sr}
= \Tr\left[W_d(p)\right].$
\begin{prop} \label{(factmom)}
If $(\cdot)_i$ denotes the lower factorial and $S(i,k)$ are Stirling numbers of second kind, then
$$E \left[ \left( {\mathcal N}_{\sr} \right)_i \right] = E \left\{ \left( \Tr \left[W_d(p) \right] \right)^i \right\}
\quad \hbox{and} \quad E \left[ \left( {\mathcal N}_{\sr} \right)^i \right] = \sum_{k=1}^i S(i,k) E \left\{ \left(\Tr\left[W_d(p)
\right] \right)^{k}\right\}.$$
\end{prop}
\begin{proof}
Conditioned factorial moments of Poisson r.v.'s are powers of the random parameter, that is $E \left[ \left( {\mathcal N}_{\sr} \right)_i \right | I_{1,\sr}, \ldots, I_{d,\sr} ]  = \left( \Tr \left[W_d(p) \right] \right)^i.$ Then factorial moments follow by taking the overall expectation. Moments follow by taking the overall expectation of $x^i = \sum_{k=1}^i S(i,k) (x)_k,$ after having replaced the
indeterminate $x$ with ${\mathcal N}_{\sr}.$
\end{proof}
The distribution of the overall photocounter can be computed as
\begin{equation}
{\mathbb P} \left({\mathcal N}_{\sr} = k \right) = \frac{1}{k!}  \sum_{i=0}^{\infty} \frac{(-1)^i}{i!} E \left\{ \left( \Tr \left[W_d(p) \right] \right)^{k+i} \right\} = \frac{1}{k!}  \sum_{i=0}^{\infty} \frac{(-1)^i}{i!} Y_{k+i}(c_1, \ldots, c_{k+i}),
\label{(distrphoto)}
\end{equation}
with $Y_k(x_1, \ldots, x_{k})$ (complete) exponential Bell polynomials \cite{Dinardo0}, see also Section 6.2, and $\{c_k\}$ cumulants
of $\Tr \left[W_d(p) \right].$ In \cite{Dinardo5}, moments of $\Tr\left[W_d(p)\right]$  have been expressed by using integer partitions \footnote{Recall that a partition of an integer $k$ is a sequence
$\lambda=(\lambda_1,\lambda_2,\ldots,\lambda_t)$, where
$\lambda_j$ are weakly decreasing integers and $\sum_{j=1}^t
\lambda_j = k$. The integers $\lambda_j$ are named {\it parts} of
$\lambda$. The {\it length} of $\lambda$ is the number of its
parts and will be denoted by $l({\lambda})$. A different
notation is $\lambda=(1^{r_1},2^{r_2},\ldots)$, where $r_j$ is the
number of parts of $\lambda$ equal to $j$ and $r_1 + r_2 + \cdots
= l({\lambda})$. We use the classical notation $\lambda \vdash
k$ with the meaning \lq\lq $\lambda$ is a partition of $k$\rq\rq. Set $m(\lambda)= (r_1, r_2, \ldots)$ and
$m(\lambda)! = r_1! r_2! \cdots.$ }
and cyclic polynomials
\footnote{The $i$-th cyclic polynomial is
${\mathcal C}_i(x_1,\ldots,x_i) = \sum_{\lambda \vdash i} d^{\prime}_{\lambda} x_1^{r_1} \cdots x_i^{r_i}$ with $d^{\prime}_{\lambda} =  i!/(1^{r_1} r_1! 2^{r_2} r_2! \cdots).$ See \cite{Dinardo5} for more details.}. Let us recall their expression to highlight their complexity with respect to cumulants. If ${\mathcal C}_k(\Sigma)$ denotes the $k$-th cyclic polynomial ${\mathcal C}_k(s_1,\ldots,s_k)$ with $s_k = \Tr[\Sigma^k],$ then
\begin{equation}
E\left\{\Tr\left[W_d(p)\right]^{k}\right\} = k! \sum_{j=0}^k \left\{ \sum_{\lambda \vdash j} \frac{(-1)^{l(\lambda)}}{m(\lambda)!} \, \Tr_{\lambda} (M \Sigma) \right\} \left\{ \sum_{\lambda \vdash k-j} \frac{p^{l(\lambda)}}{m(\lambda)!} \, {\mathcal C}_{\lambda} (\Sigma) \right\},
\label{(6)}
\end{equation}
where $\Tr_{\lambda} (M \Sigma) = \prod_{j \in \{1,2,\ldots\}} [\Tr(M \Sigma^{j-1})]^{r_j}$ and ${\mathcal C}_{\lambda} (\Sigma)
= \prod_{j \in \{1,2,\ldots\}} [{\mathcal C}_j(\Sigma)]^{r_j}.$
An algorithm to compute \eqref{(6)} is available in \cite{Dinardo5} relied on the symbolic method of moments.

The superposition of incoherent waves is fully employed when $\Tr\left[W_d(p)\right]$ is written  as
\begin{equation}
I_{1,\sr} + \cdots + I_{d,\sr} = \sum_{j=1}^d \sum_{i=1}^p |X_{i  j} + m_{i j}|^2 = \sum_{i=1}^p |X_{i  1} + m_{i 1}|^2 +
\cdots + |X_{i  d} + m_{i d}|^2,
\label{(super)}
\end{equation}
since $(|X_{i  1} + m_{i 1}|^2, \ldots, |X_{i  d} + m_{i d}|^2)$ are independent row vectors for $i=1, \ldots, p.$
In difference from cumulants, neither moments nor factorial mo\-ments of ${\mathcal N}_{\sr}$ take advantage of the decomposition on the right hand side of
equation \eqref{(super)}. Instead, conditioned cumulants linearize and are equal to
the random parameter of the overall photocounter for all non-negative integers $k,$ that is
\begin{equation}
\hbox{\rm Cum}_k\left( N_{1,\sr} + \cdots + N_{d,\sr}  \bigl| I_{1,\sr}, \ldots, I_{d,\sr} \right) = \sum_{j=1}^d \hbox{\rm Cum}_k
\left( N_{j,\sr}   \bigl| I_{j,\sr} \right) = \Tr[W_d(p)].
\label{(addcum1)}
\end{equation}
Unconditioned cumulants will be computed in the next section, since the symbolic method helps in shortening
the proofs. Here, we limit ourselves to observe that
\begin{equation}
\hbox{\rm Cum}_k\left( {\mathcal N}_{\sr} \right) = \sum_{i=1}^p \hbox{\rm Cum}_k\left[ N_{1 \, i} + \cdots + N_{d \, i} \right]
\label{(cum1)}
\end{equation}
with $N_{j \, i}$ $(i=1, \ldots, p, j=1, \ldots, d)$ the photocounter related to the $j$-th pixel and the $i$-th wave.
Approximations of the distribution in \eqref{(distrphoto)} involve cumulants of Wishart random matrices which
have a plainer expression compared with moments \eqref{(6)}, as we will show in the next section.

\section{The symbolic method of moments}

In the symbolic method of moments, an alphabet ${\mathcal A}=\{\alpha, \beta, \gamma, \ldots\}$ of indeterminates, named umbrae, is considered
and any umbra is related to a complex number sequence $\{a_k\}$  by  a suitable linear functional ${\mathbb E}.$
The functional ${\mathbb E}: {\mathbb C}[{\mathcal A}] \rightarrow {\mathbb C}$ is defined on the polynomial ring ${\mathbb C}[{\mathcal A}],$ and such that ${\mathbb E}[\alpha^k]=a_k$ for all non-negative integers $k \geq 1.$ We assume ${\mathbb E}[1]=1$ so that $a_0=1.$ The sequence $\{a_k\}$ is the sequence of moments of $\alpha$ and we say that $\{a_k\}$ is umbrally represented by $\alpha.$ Two umbrae can represent  the same sequence of moments, that is ${\mathbb E}[\alpha^k]={\mathbb E}[\gamma^k]$ for all non-negative integers $k \geq 1.$ In such a case we said that $\alpha$ is similar to $\gamma,$ in symbols $\alpha \equiv \gamma.$ The operator ${\mathbb E}$ factorizes on distinct umbrae, that is ${\mathbb E}[\alpha^i \beta^j \cdots \gamma^k] = {\mathbb E}[\alpha^i] {\mathbb E}[\beta^j] \cdots {\mathbb E}[\gamma^k]$ (uncorrelation property). A conditional evaluation has been introduced in \cite{Oliva} satisfying ${\mathbb E}[\alpha^i \beta^j \cdots \gamma^k | \alpha] = \alpha^i {\mathbb E}[\beta^j] \cdots {\mathbb E}[\gamma^k].$

Special umbrae are:
\begin{description}
\item[{\it a)}] the {\sl unity umbra} $u$ whose moments are $\{1\};$
\item[{\it b)}] the {\sl augumentation umbra} $\varepsilon$ whose moments are ${\mathbb E}[\varepsilon^k]=\delta_{0,k},$ for all non-negative integers $k$ with $\delta_{0,k}$ the Kronecker Delta;
\item[{\it c)}] the {\sl Bell umbra} whose moments are the Bell numbers \cite{Dinardo};
\item[{\it d)}] the {\sl singleton umbra} whose moments are ${\mathbb E}[\chi^k]=\delta_{1,k},$ for all non-negative integers $k.$
\end{description}
The sequence of cumulants $\{c_k\}$ of $\alpha$ is defined as\footnote{Within formal power series, equation (\ref{(61)}) holds independently from questions of convergence \cite{Stanley}.}
\begin{equation}
\sum_{k \geq 1} c_k \frac{z^k}{k!}  = \log \left( 1 + \sum_{k \geq 1} a_k \frac{z^k}{k!} \right).
\label{(61)}
\end{equation}
\begin{defn}\label{associated}
{\rm If $\lambda$ is an integer partition and $\{a_k\}$ (respectively $\{c_k\}$) is the sequence of moments (respectively cumulants) of $\alpha,$ the product $a_{\lambda} = a_1^{r_1} a_2^{r_2} \cdots$ (respectively $c_{\lambda} = c_1^{r_1} c_2^{r_2} \cdots$)
is said {\it associated to the partition} $\lambda.$}
\end{defn}
For example, if $\lambda = (1^2, 3)$ then $a_{\lambda} = a_1^2 a_3.$ Auxiliary umbrae are introduced as special symbols  representing operations among moments. For example, summations of $n$ distinct but similar umbrae $\alpha, \ldots, \alpha^{\prime}$ have moments
\begin{equation}
{\mathbb E}[(\alpha + \cdots + \alpha^{\prime})^k] = \sum_{\lambda \vdash k}  n^{l(\lambda)} \, d_{\lambda} \, c_{\lambda}, \quad
\hbox{where} \,\,\, d_{\lambda} = \frac{k!}{(1!)^{r_1} r_1! \, (2!)^{r_2} r_2! \cdots}
\label{dotoperation}
\end{equation}
and $c_{\lambda}$ is the sequence of cumulants of $\alpha$ associated to the partition $\lambda.$
We denote by $n \punt \alpha$ the auxiliary umbra representing the sequence of moments  (\ref{dotoperation}).
A generalization of $n \punt \alpha$ is the auxiliary umbra $\gamma \punt \alpha,$
obtained from $n \punt \alpha$ by replacing $n$ with $\gamma.$ If $\{g_k\}$ is umbrally represented by $\gamma,$ this replacement corresponds to replace $\{n^k\}$ with $\{g_k\}$ in (\ref{dotoperation}), that is
\begin{equation}
{\mathbb E}[(\gamma \punt \alpha)^k] = \sum_{\lambda \vdash k}  g_{l(\lambda)} \, d_{\lambda} \, c_{\lambda}.
\label{dotoperation1}
\end{equation}
Since the dot corresponds to a summation, $\gamma \punt \alpha$ represents a symbolic summation $\gamma$ times of the umbra $\alpha.$
The auxiliary umbra $\gamma \punt \alpha$ is named the dot product of $\gamma$ and $\alpha$ and is the symbolic counterpart of what we call generalized random sum.
\begin{defn}
{\rm The r.v.  with sequence of moments \eqref{dotoperation1} is named {\sl generalized random sum}.}
\end{defn}
Suitable choices of $\gamma$ and $\alpha$ correspond to suitable choices of $\{g_k\}$ and $\{c_k\}$ in (\ref{dotoperation1}) and allow us to recover moments of special (auxiliary) umbrae. For example $\beta \punt \chi \equiv \chi \punt \beta \equiv u.$ Two special dot-products
need to be mentioned separately: the $\alpha$-cumulant umbra $\chi \punt \alpha,$ representing the sequence of cumulants $\{c_k\}$ in
\eqref{(61)}, and the $\alpha$-factorial umbra $\alpha \punt \chi,$ representing the sequence $\{f_k\}$ such that
${\mathbb E}[(\alpha)_k] =f_k$ for all non-negative integers $k.$ By analogy with r.v.'s, the complex numbers $\{f_k\}$ are said factorial moments of $\alpha.$ Therefore the umbra $(\chi \punt \alpha) \punt \chi \equiv \chi \punt (\alpha \punt \chi)$ represents the sequence of factorial cumulants. The following definition states when an umbra may be replaced by a r.v.
\begin{defn}
{\rm An umbra $\alpha$ {\sl represents a r.v.} $X,$ if $\alpha$ umbrally represents the sequence of moments
$\{E[X^k]\},$ that is ${\mathbb E}[\alpha^k]=E[X^k]$ for all non-negative integers $k.$}
\end{defn}
In particular, the dot-product $n \punt \alpha$ represents the summation of $n$ i.i.d.r.v.'s. The auxiliary umbra $\gamma \punt \alpha$ represents  a generalized random sum. If moments are defined only up to some non-negative integer $k,$ then sequences of only $k$ elements are considered.

{\sl Symbolic representation of overall photocounters.}
 Assume to denote by $\omega_{\scriptscriptstyle d}$ the Wishart umbra representing the outer product of Gaussian vectors $(\Xbs_{i} + \mbs_{i})^{\dag} (\Xbs_{i} + \mbs_{i}).$ From equation \eqref{(3)}, $\Tr\left[W_d(p)\right]$ is a summation of traces of $(\Xbs_{i} + \mbs_{i})^{\dag} (\Xbs_{i} + \mbs_{i}).$ Therefore the umbra $\omega_{\scriptscriptstyle d,\sr} \equiv \omega_{\scriptscriptstyle d,1} + \cdots + \omega_{\scriptscriptstyle d,p}$ represents $\Tr\left[W_d(p)\right],$ with the subscript $i=1,\ldots,p$ corresponding to the subscript of the mean vector $\mbs_i.$  In \cite{Dinardo5}, a different symbolic representation of $\Tr\left[W_d(p)\right]$ has been provided in order to speed up the implementation of formula \eqref{(6)} and to take advantage of the non-centrality matrix. This symbolic representation does not take into account the additivity property of traces, which instead is very useful in dealing with mixed Poisson distributions. Since a Poisson r.v. with random parameter $\Lambda$ is represented by the umbra $\gamma \punt \beta,$ with $\gamma$ the umbra representing moments of $\Lambda$ \cite{Dinardo0} then the overall photocounter is represented by
\begin{equation}
\omega_{\scriptscriptstyle d,\sr} \punt \beta \equiv \omega_{\scriptscriptstyle d,1} \punt \beta + \cdots +
\omega_{\scriptscriptstyle d,p} \punt \beta
\label{(4.4)}
\end{equation}
where the right hand side of \eqref{(4.4)} is obtained from the left distributive property of the summation with respect
to the dot-product \cite{Dinardo0}. If $\mbs_i = \mbs$ for $i=1, \ldots, p$ then $\omega_{\scriptscriptstyle d,\sr} \equiv p \punt \omega_{\scriptscriptstyle d}.$

\begin{prop} \label{momoverall}
The umbra $\omega_{\scriptscriptstyle d,\sr}$ represents the sequence of factorial moments of ${\mathcal N}_{\sr}.$
\end{prop}
\begin{proof}
The sequence of factorial moments of ${\mathcal N}_{\sr}$ is represented by $\omega_{\scriptscriptstyle d,\sr} \punt \beta \punt \chi.$ The result follows by observing that $\beta \punt \chi \equiv u$ and  $\omega_{\scriptscriptstyle d,\sr} \punt u \equiv \omega_{\scriptscriptstyle d,\sr}.$
\end{proof}
Thanks to the symbolic representation of ${\mathcal N}_{\sr},$ cumulants of ${\mathcal N}_{\sr}$ can be computed by using cumulants of $\Tr\left[W_d(p)\right].$
\begin{thm} \label{(cumadd1)} $\hbox{\rm Cum}_k \left( {\mathcal N}_{\sr} \right) = \sum_{i=1}^k S(k,i) \left[ p (i-1)! \Tr(\Sigma^i) - i! \Tr(M \Sigma^{i-1}) \right].$
\end{thm}
\begin{proof}
Cumulants of ${\mathcal N}_{\sr}$ are represented by $\chi \punt (\omega_{\scriptscriptstyle d,\sr} \punt \beta) \equiv (\chi \punt \omega_{\scriptscriptstyle d,\sr}) \punt \beta.$ The result follows from \eqref{dotoperation1}, by observing that
${\mathbb E} \left\{ \left[ (\chi \punt \omega_{\scriptscriptstyle d,\sr}) \punt \beta \right]^k \right\}
= \sum_{i=1}^k S(k,i) \, {\mathbb E} \left[(\chi \punt \omega_{\scriptscriptstyle d,\sr})^i \right]$ (cf. \cite{Dinardo0}) and
${\mathbb E} \left[(\chi \punt \omega_{\scriptscriptstyle d,\sr})^i \right]$ is the $i$-th cumulant of $\Tr\left[W_d(p)\right].$
Its expression is given in \cite{Dinardo5}.
\end{proof}
From Theorem \ref{(cumadd1)}, the additivity property \eqref{(cum1)} of cumulants can be recovered since
\begin{equation}
\hbox{\rm Cum}_k \left( \Tr[W_d(1)] \right) = (k-1)! \Tr(\Sigma^k)  - k! \Tr(\mbs_{i}^{\dag} \mbs_{i} \Sigma^{k-1}) \quad i=1,\ldots,p.
\label{(cum1dim)}
\end{equation}
Factorial cumulants $\hbox{\rm FCum}_k \left( {\mathcal N}_{\sr} \right)$ of ${\mathcal N}_{\sr}$ are equal to cumulants of $\Tr\left[W_d(p)\right].$
\begin{thm} \label{(cumadd)} $\hbox{\rm FCum}_k \left( {\mathcal N}_{\sr} \right) =  p \, (k-1)! \Tr(\Sigma^k) - k! \Tr(M \Sigma^{k-1}).$
\end{thm}
\begin{proof}
Factorial cumulants of ${\mathcal N}_{\sr}$ are represented by $\chi \punt (\omega_{\scriptscriptstyle d,\sr} \punt \beta) \punt \chi \equiv \chi \punt \omega_{\scriptscriptstyle d,\sr}$ since $\beta \punt \chi \equiv u.$ Moments of $\chi \punt \omega_{\scriptscriptstyle d,\sr}$
are cumulants of $\Tr\left[W_d(p)\right].$
\end{proof}
\subsection{Randomized overall photocounting effect}
In literature on photocounting effect, the number of incoherent waves hitting the pixels has been considered deterministic. Here we assume this number described by a r.v. $P.$
\begin{defn} \label{4.7}
{\rm The randomized overall photocounting effect ${\mathcal N}_{\sY} = N_{\scriptscriptstyle{1,P}} + \cdots + N_{\scriptscriptstyle{d,P}}$ is the number of multivariate photoevents, if $P$ incoherent waves are superimposed on $d$ pixels with intensity}
\begin{equation}
I_{j, \scriptscriptstyle{P}} = \sum_{i=1}^P |X_{i j} + m_{j}|^2 \qquad \hbox{\rm for} \,\, j=1,2, \ldots, d.
\label{(7)}
\end{equation}
\end{defn}
Note that equation (\ref{(7)}) is obtained from  (\ref{(4)}),  setting $m_{i j} = m_i$ for all non-negative integers $i,$ in order to have $\Xbs_i + \mbs \sim N(\mbs, \Sigma)$ and $\mbs = (m_1, \ldots, m_d).$ From (\ref{(7)}), the random parameter of $N_{1, \scriptscriptstyle{P}} + \cdots + N_{d, \scriptscriptstyle{P}}$
may be written as
$I_{1, \scriptscriptstyle{P}} + \cdots + I_{d, \scriptscriptstyle{P}} = \sum_{i=1}^P \Tr[(\Xbs_{i} + \mbs)^{\dag} (\Xbs_{i} + \mbs)] = \Tr[W_d(P)].$

\begin{prop} \label{momoverall1}
If ${\mathcal N}_{\sY}$ is the randomized overall photocounter and $\rho$ is the umbra representing the r.v. $P$, then
$$\begin{array}{ll}
E\left[\left( {\mathcal N}_{\sY} \right)^k \right] = {\mathbb E}\left[\left(\rho \punt \, \omega_{d} \punt \, \beta \right)^k\right],
& E\left[\left( {\mathcal N}_{\sY} \right)_k\right] =  {\mathbb E}\left[\left(\rho \punt \, \omega_{d} \right)^k\right],   \\
\hbox{\rm Cum}_k \left( {\mathcal N}_{\sY} \right) = {\mathbb E}\left[ \left\{ \left( \chi \punt \rho) \punt \,(\omega_{d} \punt \beta \right)\right\}^k\right], &   \hbox{\rm FCum}_k \left( {\mathcal N}_{\sY} \right) = {\mathbb E}\left[ \left( \chi \punt \rho  \punt \, \omega_{d} \right)^k\right]
\end{array}$$
\end{prop}
\begin{proof}
From Definition \ref{4.7}, the randomized overall photocounter ${\mathcal N}_{\sY}$ is obtained from ${\mathcal N}_{\sr}$ by replacing $p$ with $P.$ From \eqref{(4.4)} the umbral counterparts of ${\mathcal N}_{\sr}$
and ${\mathcal N}_{\sY}$ are $p \punt \omega_d \punt \beta$ and $\rho \punt \omega_d \punt \beta$ respectively, since $\omega_{d, \sr} \equiv p \punt \omega_d.$ Factorial moments are represented by $\rho \punt \omega_d \punt \beta \punt \chi$ and the result follows as $\beta \punt \chi \equiv u.$  Cumulants are represented by $\chi \punt \rho \punt \omega_d \punt \beta$ and the result follows from the
associativity property. Factorial cumulants are represented by $\chi \punt \rho \punt \omega_d \punt \beta \punt \chi
\equiv \chi \punt \rho \punt \omega_d.$
\end{proof}
\begin{cor} \label{cory1} If ${\mathcal N}_{\sY}$ is the randomized overall photocounting effect, then ${\mathcal N}_{\sY} \overset{{\scriptscriptstyle d}}{=} \sum_{i=1}^P N_{i,\scriptscriptstyle{[1]}},$ where $N_{i,{\scriptscriptstyle{[1]}}}$ represents the overall photocounter of the $i$-th wave.
\end{cor}
\begin{proof}
Since $\rho  \punt (\omega_d \punt \beta) \equiv (\rho \punt \omega_d) \punt \beta,$ the result follows observing that these two umbrae represent respectively $\sum_{i=1}^P N_{i,{\scriptscriptstyle{[1]}}}$ and ${\mathcal N}_{\sY}.$ The equality in distribution follows since they both have convergent moment generating function.
\end{proof}
In the following, when no misunderstandings occur, we denote $N_{i,[1]}$ simply by $N_{[1]}.$ By using the symbolic method, moments, factorial moments and cumulants of the randomized overall photocounter ${\mathcal N}_{\sY}$ can be
easily recovered.
\begin{prop} \label{rrr} If ${\mathcal N}_{\sY}$ is the randomized overall photocounter, then
\begin{description}
\item[{\it i)}] $E \left[ \left( {\mathcal N}_{\sY} \right)^k \right]  =  \sum_{\lambda \vdash k} d_{\lambda} \, E \left[ P^{l(\lambda)} \right] \, \hbox{\rm Cum}_{\lambda} \left( N_{\scriptscriptstyle{[1]}} \right);$
\item[{\it ii)}] $E \left[ \left({\mathcal N}_{\sY} \right)_k \right] =  \sum_{\lambda \vdash k} d_{\lambda} \, E \left[ P^{l(\lambda)} \right] \,  \hbox{\rm Cum}_{\lambda} \left(\Tr[W_d(1)] \right);$
\item[{\it iii)}] $\hbox{\rm Cum}_k \left( {\mathcal N}_{\sY} \right) = \sum_{\lambda \vdash k} d_{\lambda} \, \hbox{\rm Cum}_{\l(\lambda)}(P) \,  \hbox{\rm Cum}_{\lambda} \left( N_{\scriptscriptstyle{[1]}} \right).$
\item[{\it iv)}] $\hbox{\rm FCum}_k \left( {\mathcal N}_{\sY} \right) = \sum_{\lambda \vdash k} d_{\lambda} \, \hbox{\rm Cum}_{\l(\lambda)}(P) \,  \hbox{\rm Cum}_{\lambda} \left(\Tr[W_d(1)] \right).$
\end{description}
\end{prop}
\begin{proof}
Moments {\it i)} follow  from (\ref{dotoperation1}) by replacing $\gamma$ with $\rho$ and $\alpha$ with $\omega_d \punt \beta.$ Factorial moments {\it ii)} follow  from (\ref{dotoperation1}) by replacing $\gamma$ with $\rho$ and $\alpha$ with $\omega_d.$
Cumulants {\it iii)} follow from (\ref{dotoperation1})  by replacing $\gamma$ with $\chi \punt \rho$ and $\alpha$ with $\omega_d
\punt \beta.$ Note that moments of $\chi \punt \rho$ are cumulants of $\rho.$ Factorial cumulants {\it iv)} follow from (\ref{dotoperation1}) by replacing $\gamma$ with $\chi \punt \rho$ and $\alpha$ with $\omega_d.$
\end{proof}
To compute cumulants of ${\mathcal N}_{\sY},$ cumulants of $N_{[1]}$ are necessary. They can be recovered from  \eqref{(cum1dim)} setting
$\mbs_j = \mbs.$
\section{Photocounting statistics}

When facing with sampled photocounters, two problems need to be solved: to check if the usual hypothesis of
semi-classical theory of statistical optics hold, that is to infer about the Poisson distribution of photocounting,
and to estimate its intensity field. Both tasks can be performed by using $U$-statistics \cite{Dinardo}, in a different way depending on which kind of information have been sampled.

If we have a random sample of overall photocounters $\tilde{\Nbs}^{[1]} = \left(\tilde{N}^{[1]}_1, \tilde{N}^{[1]}_2, \ldots, \tilde{N}^{[1]}_n\right)$ for each wave, the first task may be performed simply by checking the additivity property in \eqref{(cum1)}
and so by estimating cumulants. Cumulants and their products can be estimated from a random sample $\xbs = (x_1, x_2, \ldots, x_n)$ by using a family of $U$-statistics $\kappa_{\lambda}(\xbs)$ called {\it polykays} \cite{Dinardo}. If $\{c_i\}$ is the sequence of cumulants of $N^{[1]}$ (or ${\mathcal N}_{\sr}$), then $E[\kappa_{\lambda}(\xbs)] = c_1^{r_1} c_2^{r_2} \ldots.$ Polykays up to order $3$ are:
\begin{eqnarray*}
\kappa_{1}(\xbs) & = & \frac{s_{1}}{n}, \,\, \kappa_{1^2}(\xbs) = \frac{s_1^2 - s_2}{n(n-1)}, \,\, \kappa_{2}(\xbs) =
\frac{n \, s_2 - s_1^2}{n \, (n-1)}, \,\, \kappa_{1^3}(\xbs) =  \frac{s_1^3 - 3 s_1 s_2 + 2 \, s_3 }{n(n-1)(n-2)} \nonumber \\
\kappa_{1,2}(\xbs) & = & \frac{- s_1^3 + (n+1) s_1 s_2 - n s_3}{n(n-1)(n-2)}, \,\, \kappa_{3}(\xbs) = \frac{2 s_1^3 - 3 \, s_1 s_2 n - n^2 s_3}{n(n-1)(n-2)}
\end{eqnarray*}
where $s_j = \sum_{i=1}^n x_i^j$ are power sum symmetric polynomials in the sample $\xbs$ for all non-negative integers $j.$

The single index $\kappa$'s are the $k$-statistics; the multi-index $\kappa$'s are the polykays. The {\sl degree} is the sum of the subscripts, that is the integer of which $\lambda$ is a partition. The size $n$ of the sample needs to be greater than the degree.
Polykays were introduced by Fisher (see \cite{Dinardo} and references therein for more details) as \lq\lq inherited on the average\rq\rq, a property which gives to these functions a common interpretation independent of the sample size \cite{Dinardo4}. The \lq\lq inheritence\rq\rq property states that if $\ybs$ is a sub-sample of $\xbs$ obtained by simple random sampling, then $E[\kappa_{\lambda}(\ybs) | \xbs] = \kappa_{\lambda}(\xbs).$

Like cumulants, polykays enjoy of the additivity property \eqref{(cum1)}.
\begin{prop}
Polykays computed on the overall photocounting effect of $p$ waves linearize in polykays computed on the overall photocounting effect of a single wave.
\end{prop}
Also equation \eqref{(addcum1)} may be useful as a first step to verify if the underlying stochastic model is of Poisson type.
Indeed, if photocounters have been sampled for a fixed vector of intensities (for example by using Monte-Carlo methods), then
single polykays conditioned to known intensities should result approximatively constant.  Moreover, Proposition \ref{rrr} shows
that polykays are useful to compute $\hbox{\rm Cum}_k \left( {\mathcal N}_{\sY} \right)$ and $\hbox{\rm FCum}_k \left( {\mathcal N}_{\sY} \right).$

As mentioned at the beginning of the section, a different task consists in estimating the intensity field parameters from sampled photocounters. According to Proposition \ref{(factmom)}, this estimation can be easily carried out using $U$-statistics for factorial moments. $U$-statistics of factorial moments can be recovered by expressing factorial moments $\{f_k\}$ in terms of power sums
$s_j$ of the sampled overall photocounters $\tilde{\Nbs}^{\sr}.$ Again the symbolic method of moments helps in finding their
expression, as the following proposition shows.
\begin{prop} The $U$-statistic ${\mathfrak f}_k$ for the $k$-th factorial moment of ${\mathcal N}_{\sr}$ is
$${\mathfrak f}_k = \frac{1}{n} \sum_{\lambda \vdash k} d_{\lambda} \, s_{l(\lambda)} \tilde{c}_{\lambda} \quad
\hbox{with} \,\,\, \tilde{c}_{\lambda} = \prod_{i \geq 1} \left[(-1)^{i-1} (i-1)! \right]^{r_i}.$$
\end{prop}
\begin{proof}
From (\ref{dotoperation1}), factorial moments of $\alpha$ are $E[(\alpha \punt \chi)^k] = \sum_{\lambda \vdash k} d_{\lambda} a_{l(\lambda)} \tilde{c}_{\lambda}$ with $\tilde{c}_{\lambda}$ the sequence of cumulants of $\chi,$ associated to the partition $\lambda.$ The result follows since $a_{k}$ can be estimated by sample moments $s_{k}/n$ and cumulants of the singleton umbra are
$\{(-1)^{k-1} (k-1)!\},$ see \cite{Dinardo0}.
\end{proof}

A completely different scenario arises when we wish to predict photocounting effect from sampled intensities by using
equation \eqref{(distrphoto)}. In this case, we need to estimate cumulants of Wishart random matrices. The reasons are
twofold. Factorial cumulants of overall photocounters are equal to cumulants of Wishart random matrices and linearize
on outers product of $\{\Xbs_i\},$ allowing to check the independence property. Moreover, since complete Bell polynomials are easily recovered from any symbolic packages, these estimators allow us to recover also an approximation of probability distribution
(\ref{(distrphoto)}). Let us underline that in the literature on photocounting, factorial cumulants are calculated from moments of the recorded photon counts by using the classical moment conversion equations.  Here we propose a different strategy relied on {\it spectral polykays} $\tilde{\kappa}_{\lambda} \left( A \right),$ introduced in \cite{Dinardo4}. Spectral polykays are unbiased estimators of cumulants of random matrices $A.$ For Wishart random matrices, spectral polykays still have
the inheritance property and estimate cumulants normalized to the dimension \cite{Dinardo1}:
\begin{equation}
E \left\{ \tilde{\kappa}_{\lambda} \left[ W_d(p) \right] \right\} = \frac{1}{d^{l(\lambda)}} \hbox{Cum}_{\lambda} \left[ \Tr \left( W_d(p) \right) \right].
\label{(specpol)}
\end{equation}
In difference from polykays, referring to random samples of a population, spectral polykays involve spectral samples. A spectral
sample is the eigenvalue vector ${\boldsymbol e}$ of a random matrix $A,$ if the size of the sampling $n$ is equal to the order
of the matrix $A$. If $n < d$ a suitable subsample of spectral decomposition is selected. In the following, we assume $n=d.$

Spectral polykays up to order $3$ are:
\begin{eqnarray}
\tilde{\kappa}_{1}({\boldsymbol e}) & = & \frac{\Tr \left[ W_d(p) \right]}{d}, \,\, \tilde{\kappa}_{1^2}({\boldsymbol e}) = \frac{d \,  (\Tr \left[ W_d(p) \right])^2 -
\Tr \left[ W_d(p)^2 \right]}{d \, (d^2-1)},  \\
\tilde{\kappa}_{2}({\boldsymbol e})  & = &  \frac{d \, (\Tr \left[ W_d(p) \right])^2 - \Tr \left[ W_d(p)^2 \right]}{d \,(d^2-1)} \nonumber \\
\tilde{\kappa}_{1^3}({\boldsymbol e}) & = & \frac{(\Tr \left[ W_d(p) \right])^3 (d^2 - 2) - 3 \, d \, \Tr \left[ W_d(p) \right] \Tr \left[ W_d(p)^2 \right]
+ 4 \Tr \left[ W_d(p)^3 \right]}{d \, (d^2-1) \, (d^2-4)}, \nonumber \\
\tilde{\kappa}_{1,2}({\boldsymbol e}) & = & \frac{- 2 \, d \, \Tr \left[ W_d(p)^3 \right] + (d^2 + 2) \, \Tr \left[ W_d(p) \right] \Tr \left[ W_d(p)^2 \right]
- d \, (\Tr \left[ W_d(p) \right])^3}{d \, (d^2-1) \, (d^2-4)}, \nonumber \\
\tilde{\kappa}_{3}({\boldsymbol e}) & = & 2 \, \frac{2 (\Tr \left[ W_d(p) \right])^3 - 3 \, d \, \Tr \left[ W_d(p) \right] \Tr \left[ W_d(p)^2 \right] + d^2 \Tr \left[ W_d(p)^3 \right]}{d \, (d^2 - 1)\, (d^2 - 4)}. \label{(poly1bis)}
\end{eqnarray}
For completeness, the general formula to recover spectral polykays in terms of intensities is provided in Theorem
\ref{spectr}. The proof is given in \cite{Dinardo4}. The algorithm implementing this formula is available in \cite{DinardoMaple}. Permutations\footnote{A permutation $\sigma$ of $[k]$ can be decomposed into disjoint cycles $C(\sigma).$ The length of the cycle $c \in C(\sigma)$ is its cardinality, denoted by ${\mathfrak l}(c).$ The number of cycles of $\sigma$ is denoted by $|C(\sigma)|.$ Recall that a permutation $\sigma$ with $r_1$ $1$-cycles, $2$-cycles and so on is said to be of cycle class $\lambda = (1^{r_1}, 2^{r_2}, \ldots) \vdash k.$ } of cycle structure $\lambda$ are involved.
\begin{thm} \label{spectr}
If $\lambda \vdash k,$ then
$ E \left\{ \tilde{\kappa}_{\lambda} \left[ W_d(p) \right] \right\}
= \prod_{j} (j!)^{r_j} \sum_{\tau \, \omega = \sigma} {\mathfrak{Tr}}(I_d)^{-1}(\tau) E \left\{ {\mathfrak{Tr}} \left[ W_d(p) \right](\omega) \right\},$ where ${\mathfrak{Tr}}(A)(\sigma) = \prod_{c \in C(\sigma)} \Tr[A^{{\mathfrak{l}(c)}}],$ with $A$ either
the matrix identity $I_d$ either the Wishart random matrix $W_d(p),$ and ${\mathfrak{Tr}}(A)^{-1} $ is the inverse function
\footnote{The inverse function $f^{-1}(\sigma)$ of $f(\sigma)$ satisfies  $\sum_{\tau \, \omega = \sigma} f(\tau) \, f^{-1}(\omega)
= \sum_{\tau \, \omega = \sigma} f^{-1}(\tau) \, f(\omega) = \delta(\sigma)$ where $\delta(\sigma) = 1$ if $\sigma$ is the permutation identity, $0$ otherwise.} of ${\mathfrak{Tr}}(A).$
\end{thm}
\section{Multivariate photocounting effect}
A detailed description of photocounting involves the computation of joint moments and joint cumulants
\begin{equation}
E \left(N_{1, \sr}^{k_1} \cdots N_{d, \sr}^{k_d} \right) = m^{\sr}_{\kbs} \qquad
\hbox{Cum}_{\kbs} \left(N_{1, \sr}, \ldots, N_{d, \sr} \right) = c^{\sr}_{\kbs}
\label{(mult)}
\end{equation}
with $\kbs = (k_1, \ldots, k_d) \in \mathbb{N}_0^d$ and $\Nbs^{\sr} = (N_{1, \sr}, \ldots ,N_{d, \sr}).$ In order to deal with sequences (\ref{(mult)}), the symbolic method of moments has
been generalized to multi-index $\kbs$ and vectors of umbral monomials \cite{Dinardo3}. Vectors of umbral
monomials correspond to correlated random vectors when their supports\footnote{The support of an umbral polynomial $p \in R[A]$ is the set of all umbrae which occur.} are not  disjoint.  Following the notations introduced in \cite{Dinardo3}, $f(\Nbs^{\sr}, \zbs)$ denotes the moment generating function\footnote{Note that in the literature, the probability generating function of $\Nbs^{\sr}$ is erroneously called moment generating function since factorial moments are usually recovered from it.} of $\Nbs^{\sr}.$ In the following, for brevity, we referred to  the number $p$ of superimposed waves only when necessary.
\subsection{Multivariate moment symbolic method}
Let $\{\nu_1, \ldots, \nu_d\} \in {\mathbb C}[{\mathcal A}]$ a set of umbral monomials
with support not necessarily disjoint. A complex sequence $\{a_{\kbs}\},$ with
$a_{\kbs} = a_{k_1  \ldots k_d}$ and $a_{\bf 0} = 1$, is represented by the $d$-tuple $\nubs=(\nu_1,\ldots,\nu_d)$ iff
\begin{equation}
{\mathbb E}[\nubs^{\kbs}] = a_{\kbs}, \qquad \kbs \in \mathbb{N}_0^d.
\label{(multmoments)}
\end{equation}
If $\{\nu_1, \ldots, \nu_d\}$ are umbral monomials with disjoint supports then $a_{\kbs}= {\mathbb E}[\nu_1^{k_1}] \cdots {\mathbb E}[\nu_d^{k_d}]$. The elements $a_{\kbs}$ in (\ref{(multmoments)}) are called {\it
multivariate moments} of $\nubs$ and, by analogy with random vectors,
$$f(\nubs,\zbs) = 1 + \sum_{i \geq 1} \sum_{|\kbs| = i} a_{\kbs} \frac{\zbs^{\kbs}}{\kbs!}$$
is the moment generating function of $\nubs,$ with $\zbs = (z_1,  \ldots, z_d), |\kbs|=k_1 +  \cdots + k_d$ and $\kbs!=k_1! \, \cdots k_d!.$
 Two umbral $d$-tuples $\nubs_1$ and $\nubs_2$ are said to be {\it uncorrelated} if and only if ${\mathbb E}[\nubs_1^{\kbs} \, \nubs_2^{\jbs}]= {\mathbb E}[\nubs_1^{\kbs}]{\mathbb E}[\nubs_2^{\jbs}]$ for all $\kbs, \jbs \in  \mathbb{N}_0^d.$
 They are said to be similar if ${\mathbb E}[\nubs_1^{\kbs}]={\mathbb E}[\nubs_2^{\kbs}]$ for all $\kbs \in \mathbb{N}_0^d,$ in symbols $\nubs_1 \equiv \nubs_2.$ As done for the univariate case, if the sequence $\{a_{\kbs}\}$ is umbrally represented by $\nubs,$ its
 sequence $\{c_{\kbs}\}$ of multivariate cumulants satisfies
\begin{equation}
\sum_{i \geq 1} \,  \sum_{|\kbs|=i} c_{\kbs} \frac{\zbs^{\kbs}}{\kbs!} = \log \left[  f(\nubs,\zbs) - 1 \right].
\label{(genfun2)}
\end{equation}
Next definition generalizes Definition \ref{associated} to multi-index partitions.
\begin{defn}\label{associated1}
{\rm If $\lambdabs$ is a multi-index partition \footnote{ A partition of a multi-index $\lambdabs \vdash \mbs$ is a matrix $\lambdabs = (\lambda_{j t})$ of non-negative integers and with no zero columns in lexicographic order such that $\lambda_{j1}+\lambda_{j2}+\cdots =m_j$ for
$j=1,2,\ldots,d.$ As for integer partitions, the notation $\lambdabs = (\lambdabs_{1}^{r_1}, \lambdabs_{2}^{r_2}, \ldots)$
means that in the matrix $\lambdabs$ there are $r_1$ columns equal to $\lambdabs_{1}$,
$r_2$ columns equal to $\lambdabs_{2}$ and so on, with $\lambdabs_{1} <
\lambdabs_{2} < \cdots$. The multiplicity of $\lambdabs_i$ is $r_i$ and we set $\ml(\lambdabs)=(r_1, r_2,\ldots).$ The number of columns
of $\lambdabs$ is denoted by $l(\lambdabs).$} and $\{a_{\kbs}\}$ (respectively  $\{c_{\kbs}\}$) is the sequence of multivariate moments (respectively multivariate cumulants) of $\nubs,$ the product $a_{\lambdabs} = a_{\lambdabs_1}^{r_1} a_{\lambdabs_2}^{r_2} \cdots$ (respectively $c_{\lambdabs} = c_{\lambdabs_1}^{r_1} c_{\lambdabs_2}^{r_2} \cdots$)
is said {\it associated to the partition} $\lambdabs.$}
\end{defn}
\noindent
For example a multi-index partition $\lambdabs$ of $\kbs=(2,1,5)$ is $\lambdabs = (\lambdabs_1, \lambdabs_2, \lambdabs_3)$ with
$\lambdabs_1' = (0,0,1), \lambdabs_2' = (1,0,1), \lambdabs_3' = (1,1,2)$ and $a_{\lambdabs}=a_{0 \,0 \, 1} \, a_{1 \, 0 \,1}\, a_{1 \, 1 \,2}.$
If in the dot-product $\gamma \punt \alpha,$ the umbra
$\alpha$ is replaced by the $d$-tuple $\nubs,$ then equation (\ref{dotoperation1}) generalizes in
\begin{equation}
{\mathbb E}[(\gamma \punt \nubs)^{\kbs}] =  \sum_{\lambdabs \mmodels \kbs}
\frac{\kbs!}{\ml(\lambdabs)! \lambdabs!} \, g_{l(\lambdabs)} \, c_{\lambdabs}, \label{(eq:15)}
\end{equation}
where $c_{\lambdabs}$ is the product of multivariate cumulants of $\nubs$ associated to $\lambdabs.$
\begin{defn}
{\rm The r.v.  with sequence of moments \eqref{(eq:15)} is named {\sl multivariate generalized random sum}.}
\end{defn}
Suitable choices of $\gamma$ and $\nubs$ correspond to suitable choices of $\{g_k\}$ and $\{c_{\kbs}\}$ in (\ref{(eq:15)}) and allow us to recover moments of special (auxiliary) umbrae. A special dot-product is the $\nubs$-cumulant umbra $\chi \punt \nubs,$ representing the sequence of cumulants $\{c_{\kbs}\}$ in \eqref{(eq:15)}.  The following definition states when a $d$-tuple of umbral monomials
may be replaced by a random vector.
\begin{defn}
{\rm A $d$-tuple $\nubs$ of umbral monomials  represents a random vector  $\Xbs,$ if $\nubs$ umbrally represents the sequence of
multivariate moments $\{E[\Xbs^{\kbs}]\},$ that is ${\mathbb E}[\nubs^{\kbs}]=E[\Xbs^{\kbs}]$ for all $\kbs \in \mathbb{N}_0^d.$}
\end{defn}
As in the univariate case, the dot-product $\gamma \punt \nubs$ represents  a sum of random vectors indexed by a not necessarily univariate
integer-value r.v, what we have called multivariate generalized random sum. In particular $\chi \punt \nubs$ represents the $\nubs$-cumulant
umbra.
\subsection{Computations of joint photocounters}
Complete Bell polynomials $Y_k(x_1, \ldots, x_{k})$ in \eqref{(distrphoto)} are
\begin{equation}
Y_k(x_1, \ldots, x_{k}) = \sum_{i=1}^k  \sum_{\lambda \vdash k, \, l(\lambda)=i} d_{\lambda} \,\, x_1^{r_1} x_2^{r_2} \cdots
\label{(bellpol)}
\end{equation}
Joint moments and joint cumulants (\ref{(mult)}) can be computed by using suitable generalizations of complete Bell polynomials
$Y_k(x_1, \ldots, x_{k})$ with the indeterminates $\{x_1, \ldots, x_{k}\}$ replaced by umbrae.  More precisely, let us consider the auxiliary umbra $\gamma \punt \beta \punt \alpha,$ that is the summation $\gamma$ times of $\beta \punt \alpha.$ The auxiliary umbra $\beta \punt \alpha$ represents
a compound Poisson r.v. of parameter $1,$ that is a summation $N$ times of a r.v. represented by $\alpha,$ with $N \simeq \hbox{Po}(1).$ Moments of $\gamma \punt \beta \punt \alpha,$ computed by means of equation (\ref{dotoperation1}), result to be a first generalization of complete Bell polynomials \eqref{(bellpol)}. Indeed, when $\alpha$ is replaced by $\beta \punt \alpha$ in equation (\ref{dotoperation1}), we have
\cite{Dinardo}
\begin{equation}
{\mathbb E}[(\gamma \punt \beta \punt \alpha)^k] =  \sum_{\lambda \vdash k} d_{\lambda} \, g_{l(\lambda)} \,a_{\lambda}
\label{(Bellpol1)}
\end{equation}
with $\{a_k\}$ and $\{g_k\}$ umbrally represented by the umbra $\alpha$ and $\gamma$ respectively. The generalization to the multivariate case of equation (\ref{(Bellpol1)}) is obtained by replacing $\alpha$ with the $d$-tuple $\nubs=(\nu_1,\ldots,\nu_d)$
\begin{equation}
{\mathbb E}[ (\gamma \punt \beta \punt \nubs)^{\kbs}] =  \sum_{\lambdabs \mmodels \kbs} \,
\frac{\kbs!}{\ml(\lambdabs)! \lambdabs!} \, g_{l(\lambdabs)} \, a_{\lambdabs}, \label{(eq:16)}
\end{equation}
where $a_{\lambdabs}$ is the product of multivariate moments of $\nubs$ associated to $\lambdabs.$ More general expressions
of equation (\ref{(eq:16)}) correspond to moments of
$\gamma_1 \punt \beta \punt \nubs_1 + \cdots + \gamma_d \punt \beta \punt \nubs_d$ with  $\{\nubs_1, \ldots, \nubs_d\}$ $d$-tuples of umbral monomials \cite{Dinardo3}. Set
\begin{equation}
(\gamma_1 \punt \beta \punt \nubs_1 + \cdots + \gamma_d \punt \beta \punt \nubs_d)^{\kbs} = {\mathcal Y}_{\kbs}^{(\nubs_1, \ldots, \nubs_d)}(\gamma_1, \ldots, \gamma_{d}).
\label{(multbern)}
\end{equation}
The polynomials ${\mathcal Y}_{\kbs}^{(\nubs_1, \ldots, \nubs_d)}(\gamma_1, \ldots, \gamma_{d})$ are said {\sl generalized complete Bell polynomials}. For $d=1$ we recover ${\mathcal Y}_{\kbs}^{(\nubs)}(\gamma) = (\gamma \punt \beta \punt \nubs)^{\kbs}.$ By using a suitable generalization of multinomial expansion \cite{Dinardo3}, the $\kbs$-th moment of $(\gamma_1 \punt \beta \punt \nubs_1 + \cdots + \gamma_d \punt \beta \punt \nubs_d)$ is
\begin{equation}
{\mathbb E}\left[{\mathcal Y}_{\kbs}^{(\nubs_1, \ldots, \nubs_d)}(\gamma_1, \ldots, \gamma_{d}) \right] = \sum_{(\ibs_1, \ldots, \ibs_d): \sum_{j=1}^d \ibs_j = \kbs}  \binom{\kbs}{\ibs_1, \ldots, \ibs_d} {\mathbb E} \left[ (\gamma_1 \punt \beta \punt \nubs_1)^{\ibs_1} \cdots (\gamma_d \punt \beta \punt \nubs_d)^{\ibs_d} \right].
\label{(expansion)}
\end{equation}
In the following, let us denote by $\iota_{\scriptscriptstyle j}$ the umbra representing the intensity $I_{\scriptscriptstyle j}$ for $j=1,\ldots,d$ and by $\iotabs$ the corresponding $d$-tuple $\iotabs \equiv (\iota_1, \ldots, \iota_d).$
\begin{thm} \label{teorema1} If $\kbs \in \mathbb{N}_0^d,$ then
$m_{\kbs}  = {\mathbb E}\left[{\mathcal Y}_{\kbs}^{(\ubs_{\scriptscriptstyle 1}, \ldots, \ubs_{\scriptscriptstyle d})} \left(\iota_{\scriptscriptstyle 1}, \ldots, \iota_{\scriptscriptstyle d} \right) \right]$
with
\begin{equation}
\ubs_j=( \varepsilon, \ldots, \!\!\! \underbrace{u}_{j\hbox{\tiny{-th place}}} \! ,\!\ldots, \varepsilon), \quad
\hbox{\rm for $j=1,\ldots,d.$}
\label{(uvector)}
\end{equation}
\end{thm}
\begin{proof}
From \eqref{(uvector)}, we have $f(\ubs_{\scriptscriptstyle j}, \zbs) = e^{z_j}$ for $j=1,2,\ldots,d$ and
\begin{equation}
f(\Nbs, \zbs) =   E \left\{ \prod_{j=1}^d E \bigl[ \exp \left( N_{{\scriptscriptstyle j}} \, z_j \right) \, | \, I_{{\scriptscriptstyle j}} \bigr] \right\} = {\mathbb E} \left( \exp \left[ \sum_{j=1}^d  \iota_{{\scriptscriptstyle j}} \left\{f(\ubs_{\scriptscriptstyle i}, \zbs) - 1\right\}  \right] \right).
\label{(multgf)}
\end{equation}
The result follows by observing that the right hand side of (\ref{(multgf)}) is the moment generating function of $\iota_{\scriptscriptstyle 1} \, \punt \beta \punt \ubs_{\scriptscriptstyle 1} + \cdots + \iota_{\scriptscriptstyle d} \, \punt \beta \punt \ubs_d.$
\end{proof}
Next corollary gives the explicit expression of joint moments of photocounters.
\begin{cor} \label{6.4} If $\kbs = (k_1, \ldots, k_d) \in \mathbb{N}_0^d$ then $m_{\kbs} = E \left\{ \prod_{j=1}^d \left[ \sum_{i=1}^{k_j} S(k_j,i) \, I_{j}^i \right] \right\}.$
\end{cor}
\begin{proof}
Let us consider the expansion \eqref{(expansion)} with $(\nubs_{\scriptscriptstyle 1}, \ldots, \nubs_{\scriptscriptstyle d})$ replaced by $(\ubs_{\scriptscriptstyle 1}, \ldots, \ubs_{\scriptscriptstyle d})$ and $(\gamma_{\scriptscriptstyle 1}, \ldots, \gamma_{\scriptscriptstyle d})$ replaced by $(\iota_{\scriptscriptstyle 1}, \ldots, \iota_{\scriptscriptstyle d}).$
Since
\begin{equation}
{\mathbb E} \left[\ubs_{i}^{\jbs} \right] = \left\{ \begin{array}{cl}
0, & \hbox{if $\jbs$ is such that} \,\, j_{k} \ne 0 \,\, \hbox{for} \, k=1,2,\ldots,i-1,i+1, \ldots d, \\
1, & \hbox{otherwise}
\end{array} \right.
\label{(multmom)}
\end{equation}
among the vectors $(\ibs_1, \ldots, \ibs_d)$ giving not zero contributions in (\ref{(expansion)}), there are those satisfying
$\left(\ibs_{j}\right)_{i}=k_j \delta_{j,i}$ for $j=1,2,\ldots,d.$ In this case $\binom{\kbs}{\ibs_1, \ldots, \ibs_d}=1.$
Since
\begin{equation}
{\mathbb E} \left[ (\iota_{\scriptscriptstyle 1} \, \punt \beta \punt \ubs_{\scriptscriptstyle 1})^{\ibs_{\scriptscriptstyle 1}} \cdots (\iota_{{\scriptscriptstyle d}} \, \punt \beta \punt \ubs_{\scriptscriptstyle d})^{\ibs_{\scriptscriptstyle d}} \right] = {\mathbb E} \left\{ {\mathbb E} \left[ (\iota_{{\scriptscriptstyle 1}} \, \punt \beta \punt \ubs_{\scriptscriptstyle 1})^{\ibs_{\scriptscriptstyle 1}} \cdots (\iota_{{\scriptscriptstyle d}} \, \punt \beta \punt \ubs_{\scriptscriptstyle d})^{\ibs_{\scriptscriptstyle d}} \, | \,
\iota_{\scriptscriptstyle 1}, \ldots, \iota_{\scriptscriptstyle d}  \right] \right\}
\label{(eq:18)}
\end{equation}
the result follows by multiplying ${\mathbb E}[(\iota_{\scriptscriptstyle j} \, \punt \beta \punt \ubs_{\scriptscriptstyle j})^{\ibs_{\scriptscriptstyle j}} \, | \, \iota_{{\scriptscriptstyle j}}] = {\mathbb E}[(\iota_{{\scriptscriptstyle j}} \,
\punt \beta)^{k_j} \, | \, \iota_{{\scriptscriptstyle j}}] = \sum_{t=1}^{k_j} S(k_j,t) \, \iota_{{\scriptscriptstyle j}}^{t}$ for $j=1,2,\ldots,d$ and getting the overall expectation of the product.
\end{proof}
\begin{cor} \label{momzero} For $\kbs = (k_1, \ldots, k_d) \in \mathbb{N}_0^d,$ if there exists $j \in \{1,2,\ldots,d\}$ such that $k_j=0$ then $m_{\kbs}=0.$
\end{cor}
\begin{rem} \label{jointphotocounter}
{\rm  To compute $m_{\kbs}$ by using Theorem \ref{teorema1}, a symbolic procedure is available on demand.
This consists in expanding the product in the right hand side of \eqref{(expansion)} and then replacing occurrences of $\iota_{{\scriptscriptstyle 1}}^{k_1} \cdots \iota_{{\scriptscriptstyle d}}^{k_d}$ with $E\left[I_{1}^{k_1} \cdots I_{d}^{k_d}\right].$
The algorithm {\tt nCWishart} \cite{Dinardo5} allows us to compute $E\left[I_{1}^{k_1} \cdots I_{d}^{k_d}\right].$
Indeed, this algorithm allows us to compute general joint moments
\begin{equation}
E\left\{ \Tr \left[ W_d(p) \, H_1 \right]^{k_1} \cdots \Tr \left[ W_d(p) \, H_d \right]^{k_d} \right\}
\qquad \hbox{with} \quad H_1, \ldots, H_d \in {\mathbb C}^{d \times d}.
\label{(joint)}
\end{equation}
Joint moments $E\left[I_{1}^{k_1} \cdots I_{d}^{k_d}\right]$ can be recovered from (\ref{(joint)}) by choosing
$$(H_i)_{s\,t} = \left\{ \begin{array}{cl}
1, & \hbox{if $i=s=t,$} \\
0, & \hbox{otherwise,}
\end{array} \right. \qquad s,t=1,2,\ldots, d, \,\, \hbox{and} \,\, i=1,2,\ldots,d.$$}
\end{rem}
Denote by ${\mathfrak f}_{\kbs}$ the $\kbs$-th multivariate factorial moment of $\Nbs.$ As the following theorem states,
also multivariate factorial moments can be expressed via generalized complete Bell polynomials.

\begin{thm} \label{teorema2} If $\kbs \in \mathbb{N}_0^d,$ then
${\mathfrak f}_{\kbs}  = {\mathbb E}\left[{\mathcal Y}_{\kbs}^{(\chibs_1, \ldots, \chibs_d)} \left(\iota_{\scriptscriptstyle 1}, \ldots, \iota_{\scriptscriptstyle d} \right) \right]$
with
$$\chibs_j=( \varepsilon, \ldots, \!\!\! \underbrace{\chi}_{j\hbox{\tiny{-th place}}} \! ,\!\ldots, \varepsilon), \quad
\hbox{\rm for $j=1,\ldots,d.$}$$
\end{thm}
\begin{proof}
By using the same arguments employed in the proof of Theorem \ref{teorema1}, the moment generating function of the factorial moments of $\Nbs$ is
\begin{equation}
f(\Nbs, \zbs) \bigl|_{z_i=\log(1+w_i) \atop i=1,2,\ldots,d}  =  E \left\{ \prod_{j=1}^d E \bigl[ \exp \left( I_{{\scriptscriptstyle j}} \, w_j \right) \, | \, I_{{\scriptscriptstyle j}} \bigr] \right\} = {\mathbb E} \left( \exp \left[ \sum_{j=1}^d   \iota_{\scriptscriptstyle j} \left\{f(\chibs_i, \wbs) - 1\right\}  \right] \right)
\label{mmm}
\end{equation}
where $f(\chibs_i, \wbs) = 1 + w_i$ and $\wbs=(w_1, \ldots, w_d).$ The result follows by observing that the right hand side of (\ref{mmm}) is the moment generating function of $\iota_{\scriptscriptstyle 1} \, \punt \beta \punt \chibs_1 + \cdots + \iota_{\scriptscriptstyle d} \, \punt \beta \punt \chibs_d.$
\end{proof}
\begin{cor} \label{chibf} If $\kbs = (k_1, \ldots, k_d) \in \mathbb{N}_0^d,$ then ${\mathfrak f}_{\kbs} =  E \left( I_{1}^{k_1} \cdots I_{d}^{k_d} \right).$
\end{cor}
\begin{proof}
From Theorem \ref{teorema2} and equation \eqref{(multbern)} we have ${\mathfrak f}_{\kbs} = {\mathbb E} \left[ \left( \iota_{\scriptscriptstyle 1} \punt \ubs_1 + \cdots + \iota_{\scriptscriptstyle d} \punt \ubs_d \right)^{\kbs} \right],$
since $\beta \punt \chibs_j \equiv \ubs_j$ for $j=1,\ldots,d.$ Then by using the multinomial expansion
\eqref{(expansion)},
we have
\begin{equation}
{\mathbb E} \left[ \left( \iota_{\scriptscriptstyle 1} \punt \ubs_1 + \cdots + \iota_{\scriptscriptstyle d} \punt \ubs_d \right)^{\kbs} \right] = \sum_{(\ibs_1, \ldots, \ibs_d): \sum_{j=1}^d \ibs_j = \kbs}  \binom{\kbs}{\ibs_1, \ldots, \ibs_d} {\mathbb E} \left[ (\iota_{\scriptscriptstyle 1} \punt \ubs_1)^{\ibs_1} \cdots (\iota_{\scriptscriptstyle d} \punt \ubs_d)^{\ibs_d} \right].
\label{(expansion1)}
\end{equation}
In order to evaluate products on the right hand side of (\ref{(expansion1)}), equation (\ref{(eq:16)})
has to be employed. Since $\iota_{\scriptscriptstyle j} \punt \ubs_j \equiv \iota_{\scriptscriptstyle j} \punt \beta \punt \chi \punt \ubs_j$ and $\chi \punt \ubs_j \equiv \chibs_j,$ in evaluating ${\mathbb E} \left[ \prod_{j=1}^d (\iota_{\scriptscriptstyle j}
\punt \ubs_j)^{\ibs_j} \right]$ only joint products $\iota_{\scriptscriptstyle 1}^{k_1} \cdots \iota_{\scriptscriptstyle d}^{k_d}$ gives contribution, from which the result follows.
\end{proof}
Joint cumulants can be computed by using the additivity property on waves.
\begin{thm} \label{cumthm} If $\kbs \in \mathbb{N}_0^d,$ then
\begin{equation}
\hbox{\rm Cum}_{\kbs}(\Nbs^{\sr}) =  \kbs! \sum_{i=1}^p  \sum_{\lambdabs \in P_{\kbs}}
\frac{(-1)^{l(\lambdabs)-1}[l(\lambdabs)-1]!}{\ml(\lambdabs)! \lambdabs!} \, \prod_{\lambdabs_s}  E \left[ \prod_{j=1}^d
\left[ \sum_{t=1}^{(\lambdabs_s)_j} S((\lambdabs_s)_j,t) | X_{i j} + m_{i j}|^{2t} \right] \right]^{r_s},
\label{(2cum)}
\end{equation}
where $P_{\ibs} = \{\lambdabs = (\lambdabs_1^{r_1}, \lambdabs_2^{r_2}, \ldots) \vdash \kbs: (\lambdabs_s)_j \ne 0, \forall j=1, \ldots,d, s=1,2,\ldots \}.$
\end{thm}
\begin{proof}
First, let us prove the additivity property
\begin{equation}
\hbox{\rm Cum}_{\kbs}(\Nbs^{\sr}) = \sum_{i=1}^p \hbox{\rm Cum}_{\kbs}(\Nbs_{\scriptscriptstyle i}^{\scriptscriptstyle [1]}),
\label{(addcum)}
\end{equation}
where $\hbox{\rm Cum}_{\kbs}(\Nbs_{\scriptscriptstyle i}^{\scriptscriptstyle [1]})$ is the $\kbs$-th cumulant of the multivariate photocounter of the $i$-th wave with $i=1, \ldots,p.$  Indeed, denote by $\gamma_{i j}$ the umbral monomial representing $|X_{i j} + m_{i j}|^2$ in (\ref{(4)}). Then $\iota_{\scriptscriptstyle 1} \punt \beta \punt \ubs_{\scriptscriptstyle 1} + \cdots + \iota_{{\scriptscriptstyle d}} \punt \beta \punt \ubs_{\scriptscriptstyle d} = \sum_{j=1}^d (\gamma_{1 \, j} + \cdots + \gamma_{p \, j}) \punt \beta \punt \ubs_j.$ Since for fixed $j$ the umbral monomials $\{\gamma_{1 \, j}, \ldots, \gamma_{p \, j} \}$ are uncorrelated, equation \eqref{(addcum)}
follows by observing that $\sum_{j=1}^d (\gamma_{1 \, j} + \cdots + \gamma_{p \, j}) \punt \beta \punt \ubs_j \equiv \sum_{j=1}^d (\gamma_{1 \, j} \punt \beta \punt \ubs_j + \cdots + \gamma_{p \, j} \punt \beta \punt \ubs_j) = \sum_{i=1}^p (\gamma_{i \, 1} \punt \beta \punt \ubs_1 + \cdots + \gamma_{i \, d} \punt \beta \punt \ubs_d).$ Since $(\gamma_{i \, 1} \punt \beta \punt \ubs_1 + \cdots + \gamma_{i \, d} \punt \beta \punt \ubs_d)$ denotes the multivariate photocounter from $d$ pixels hit by the $i$-th wave with $i=1, \ldots,p,$ the result follows from the additivity property of multivariate cumulants. To get equation \eqref{(2cum)}, the explicit expression of
$\hbox{\rm Cum}_{\kbs}(\Nbs_{\scriptscriptstyle i}^{\scriptscriptstyle [1]})$ has to be computed by evaluating
the $\kbs$-th moment of $\chi \punt (\gamma_{i \, 1} \punt \beta \punt \ubs_1 + \cdots + \gamma_{i \, d} \punt \beta \punt \ubs_d).$
The $\kbs$-th moment of $\chi \punt \nubs$ is \cite{Dinardo3}
$$
{\mathbb E}[(\chi \punt \nubs)^{\kbs}] =  \sum_{\lambdabs \mmodels \kbs}
\frac{\kbs!}{\ml(\lambdabs)! \lambdabs!} \, (-1)^{l(\lambdabs)-1} [l(\lambdabs)-1]! \, a_{\lambdabs},
$$
with $a_{\lambdabs}$ the product of multivariate moments of $\nubs$ associated to $\lambdabs.$ The result follows by replacing $a_{\lambdabs}$ with the product of multivariate moments of  $(\gamma_{i \, 1} \punt \beta \punt \ubs_1 + \cdots + \gamma_{i \, d} \punt \beta \punt \ubs_d)$ associated to $\lambdabs.$ Moments of $(\gamma_{i \, 1} \punt \beta \punt \ubs_1 + \cdots + \gamma_{i \, d} \punt \beta \punt \ubs_d)$ are given in Corollary \ref{6.4} with $\iota_j$ replaced by $\gamma_{i \, j}$.
\end{proof}
Within estimation, the importance of Theorem \ref{cumthm} relies on the circumstance that if estimators of $\hbox{\rm Cum}_{\kbs}(\Nbs^{\sr})$ linearizes, according to \eqref{(addcum)} the underlying distribution of photocounters can be assumed of mixed multivariate Poisson type. Unbiased estimators of $\hbox{\rm Cum}_{\kbs}(\Nbs_{\scriptscriptstyle i}^{\scriptscriptstyle [1]})$ and $\hbox{\rm Cum}_{\kbs}(\Nbs^{\sr})$ can be computed  using multivariate polykays \cite{Dinardo}.
\begin{thm} \label{distrb}
The multivariate Poisson-Mandel transform admits the following expansion in series
$$
{\mathbb P}(\Nbs = \kbs) = \frac{1}{\kbs!} \left(1 + \sum_{i \geq 1} \,  \sum_{|\jbs|=i} E[\Ibs_{\jbs + \kbs}] \frac{(-1)^{|\jbs|}}{\jbs!}\right) = \frac{1}{\kbs!} \left(1 + \sum_{i \geq 1} \,  \sum_{|\jbs|=i} \frac{(-1)^{|\jbs|}}{\jbs!}
\sum_{\lambdabs \mmodels \kbs + \ibs} \frac{(\jbs+\kbs)!}{\ml(\lambdabs)! \lambdabs!} \, c_{\lambdabs}
\right)
$$
with $c_{\lambdabs}$ the product of multivariate cumulants of $\Ibs$ associated to $\lambdabs.$
\end{thm}
\begin{proof}
The first equality follows from \eqref{(2)} by observing that
$$\exp\{-(I_1 + \cdots + I_d)\} =  1 + \sum_{i \geq 1} \,  \sum_{|\jbs|=i} \Ibs^{\jbs} \frac{(-1)^{|\jbs|}}{\jbs!}.$$
The second equality follows from equation \eqref{(eq:15)} with $g_{l(\lambdabs)}$ replaced by $1,$ since $\nubs \equiv (\beta \punt \chi) \punt \nubs.$
\end{proof}
If sampled intensities are available, multivariate polykays for $\Ibs$ can be employed to approximate $c_{\lambdabs}$ in
Theorem \ref{distrb} and also factorial cumulants of $\Nbs,$ as the following theorem shows.
\begin{thm}
$\hbox{\rm FCum}_{\kbs} \left(\Nbs \right) = c_{\kbs}.$
\end{thm}
\begin{proof}
Denote by $K(\Nbs,\zbs)$ the cumulant generating function of $\Nbs,$ with $K(\Nbs,\zbs) = \log[f(\Nbs,\zbs)].$ The factorial cumulant generating function is
$$K(\Nbs, \zbs) \bigl|_{z_i=\log(1+w_i) \atop i=1,2,\ldots,d} = \log[f(\Nbs,\zbs)] \bigl|_{z_i=\log(1+w_i) \atop i=1,2,\ldots,d} = \log [f(\iotabs,\zbs)],$$
where the last equality follows from Corollary \ref{chibf}.
\end{proof}
\section{Conclusions and open problems}
This paper has introduced the symbolic method of moments as an efficient tool to deal with photon statistics.
Instead of using factorial moments, as usually proposed in the literature, cumulants have been employed due to their additivity property
which simplifies the inference on the process as well as factorial cumulants. Moreover these sequences have a simpler expression for Wishart random matrices whose traces are the intensities of photocounting. In addition, the algorithms available for computing unbiased estimators of cumulants, that are polykays and multivariate polykays, have efficient implementation.

Several open problems arise from the methodology suggested in this article. As investigated in \cite{Letac}, there is a connection between Wishart random matrices and natural exponential families. So distributions of photocounters could be studied by using properties of natural exponential families. A first attempt in this direction has been presented in \cite{Chatelain} and \cite{Ferrari}, where only special cases are considered.  A more general development, involving cumulants, could turn to be useful in finding sufficient photon statistics relied on maximum likelihood method. The use of the symbolic method of moments within these applications is currently under investigation.

A further development consists in studying photocounting vs. time, through a compound Poisson stochastic process. The so-called compensated version consists in normalizing the process with respect to its temporal mean. Cumulants play a fundamental role in
dealing with compensated Poisson process \cite{Peccati} and we believe that the techniques here introduced should be fruitfully applied.



\begin{thebibliography}{99}
%
\bibitem{Aime}
Aime C., Soummer R. (2004) \emph{Influence of speckle and Poisson noise on exoplanet detection with a coronagraph.}
in {\tt EUSIPCO'04} (L. Torres, E. Masgrau and M. A. Lagunas, eds.), 509--512, Vienna.
Elsevier.
%
\bibitem{Chatelain1}
Chatelain F., Ferrari A., Tourneret J.~Y. (2006) \emph{Parameter estimation for multivariate mixed Poisson distributions}. in
{\tt Proc. ICASSP'06}, {\bf 6}, 17--20.
%
\bibitem{Chatelain}
Chatelain F., Lambert-Lacroix S., Tourneret J.~Y. (2009) \emph{Pairwise likelihood estimation for multivariate mixed Poisson models
generated by Gamma intensities}. Stat. Comput. {\bf 19}, 283--301.
%
\bibitem{Debashis}
Debashis, P., Alexander A. (2014) \emph{Random matrix theory in statistics: a review}. Jour. Stat. Plan. Inf. {\bf 150}, 1--29.
%
\bibitem{Dinardo1}
Di Nardo E. (2014) \emph{On a symbolic representation of non-central Wishart random matrices with applications}. Jour. Mult. Anal.
{\bf 125}, 121–-135.
%
\bibitem{Dinardo5}
Di Nardo E., Guarino G. (2013) \emph{A new algorithm for computing moments of complex non-central Wishart distributions.} Worksheet Maple Software: http://www.maplesoft.com/applications/view.aspx?SID=143890.
%
\bibitem{DinardoMaple}
Di Nardo E., Guarino G. (2014) \emph{Spectral k-statistics.} Worksheet Maple Software: http://www.maplesoft.com/applications/view.aspx?SID=153618.
%
\bibitem{Dinardo3}
Di Nardo E., Guarino G., Senato D. (2008) \emph{A new algorithm for computing the multivariate Fa\'a di Bruno's formula.} Appl. Math. Comp. {\bf 217}, No. 13, 6286--6295.
%
\bibitem{Dinardo}
Di Nardo E., Guarino G., Senato D. (2009) \emph{A new method for fast computing unbiased estimators of cumulants}.
Stat. Comput. {\bf 19}, 155-–165.
%
\bibitem{Dinardo4}
Di Nardo E., McCullagh P., Senato D. (2013) \emph{Natural statistics for spectral samples}. Ann. Stat. {\bf 41}, No. 2, 982--1004.
%
\bibitem{Oliva}
Di Nardo E., Oliva I. (2013) \emph{A new family of time-space harmonic polynomials with respect to L\'evy processes}. Ann. Mat. Pura Appl.
{\bf 192}, No. 5, 917--929.
%
\bibitem{Dinardo0}
Di Nardo E., Senato D. (2001) \emph{Umbral nature of the Poisson random variables}. Algebraic Combinatorics and Computer science: a tribute to Gian-Carlo Rota. (eds. H. Crapo, D. Senato) Springer-Verlag, 245--266.
%
\bibitem{Ferrari}
Ferrari A., Letac G., Tourneret J.~Y. (2004) \emph{Multivariate Mixed Poisson Distributions}, in {\tt EUSIPCO'04}
(Hlawatsch F., Matz G., Rupp M., and Wistawel B., eds.) Elsevier, 1067--1070.
%
\bibitem{Figer}
Figer D. F., Lee J., Hanold B. J., Aull B. F., Gregory J. A., Schuette D. R. (2011) \emph{A photon-counting detector for exoplanet missions},
in {\tt Proc. SPIE 8151}.
%
\bibitem{Goodman2}
Goodman J.~W. (2007) \emph{Speckle Phenomena in Optics}. Ben Roberts $\&$ Company Publishers.
%
\bibitem{Letac}
Letac G., Massam H. (2008) \emph{The noncentral Wishart as an exponential family and its moments}. Jour. Mult. Anal. {\bf 99},
1393--1417.
%
\bibitem{Letac1}
Letac, G. Massam, H. (2004) \emph{All Invariant Moments of the Wishart Distribution.} Scand. J. Statist. {\bf 31}, No. 2,
295--318.
%
\bibitem{McCullagh}
McCullagh, P. (1987) \emph{Tensor methods in statistics.} Chapman and Hall, 285 pp.
%
\bibitem{Mordovina}
Mordovina U., Emary C. (2013) \emph{Full-counting statistics of random transition-rate matrices.} Phys. Rev. E {\bf 88}, 062148.
%
\bibitem{Muller}
Muller J. D. (2004) \emph{Cumulant Analysis in Fluorescence Fluctuation Spectroscopy.} Biophysical Jour. {\bf 86},
3981--3992.
%
\bibitem{Speicher}
Nica A. , Speicher R. (2006) \emph{Lectures on the Combinatorics of Free Probability}, London
Mathematical Society, Lecture Note Series, {\bf 335} Cambridge University Press.
%
\bibitem{O'Connor}
O'Connor P., Gehlen J., Heller E.~J. (1987) \emph{Properties of random superpositions of plane wave}. Phys. Rev. Lett. {\bf 58}, 1296--1299.
%
\bibitem{Peccati}
Peccati G., Taqqu M. (2011) \emph{Wiener Chaos: Moments, Cumulants and Diagrams: A survey with Computer Implementation.}
Bocconi $\&$ Springer Series.
%
\bibitem{Saw}
Saw J.~G. (1973) \emph{Expectation of elementary symmetric functions of a Wishart matrix.} Ann. Stat. {\bf 1}, 580--582
%
\bibitem{Shah}
Shah B.~K., Khatri C.G. (1974) \emph{Proof of conjectures about the expected values of the elementary symmetric functions of a noncentral Wishart matrix.} Ann. Stat. {\bf 2}, 833--836.
%
\bibitem{Stanley}
Stanley R.~P. (2012) \emph{Enumerative combinatorics}. Vol. I. Wadsworth $\&$ Brooks/Cole Advanced Books $\&$ Software.
%
\bibitem{Waal}
De Waal D.~J. (1972) \emph{On the expected values of the elementary symmetric functions of a noncentral Wishart matrix.} Ann. Math. Stat.
{\bf 43}, 344--347.
%
\bibitem{Tourneret}
Tourneret J.~Y., Ferrari A., Letac G. (2005) \emph{The noncentral Wishart distribution: Properties and applications to speckle imaging}, in {\tt Proc. IEEE Workshop on Stat. Signal Proc.} 1856–-1860.
%
\bibitem{Wu}
Wu B., Singer R.H., Mueller J.D. (2013) \emph{Time-integrated fluorescence cumulant analysis and its application in living cells.}, Methods Enzymol. {\bf 518}, 99 -- 119.
%
\end{thebibliography}
\end{document}